\documentclass[11pt]{article}

\usepackage[letterpaper,margin=1.15in]{geometry}
\usepackage{microtype}
\usepackage{amsmath,amssymb,amsthm,mathtools}
\usepackage{booktabs}
\usepackage{multirow}
\usepackage{float}
\usepackage{graphicx}
\usepackage{tikz}
\usetikzlibrary{arrows.meta,calc,positioning,shapes.geometric}
\usepackage{listings}
\usepackage{xcolor}
\definecolor{okblue}{HTML}{0072B2}
\definecolor{okvermil}{HTML}{D55E00}
\definecolor{okgreen}{HTML}{009E73}
\definecolor{codegreen}{rgb}{0,0.6,0}
\definecolor{codegray}{rgb}{0.5,0.5,0.5}
\definecolor{codepurple}{rgb}{0.58,0,0.82}
\definecolor{backcolour}{rgb}{0.96,0.96,0.94}
\lstdefinestyle{mystyle}{
  backgroundcolor=\color{backcolour},commentstyle=\color{codegreen},
  keywordstyle=\color{magenta},numberstyle=\tiny\color{codegray},
  stringstyle=\color{codepurple},basicstyle=\ttfamily\footnotesize,
  breaklines=true,captionpos=b,keepspaces=true,numbers=left,
  numbersep=5pt,showspaces=false,showstringspaces=false,tabsize=2}
\usepackage{natbib}
\usepackage[colorlinks=true,linkcolor=blue!55!black,citecolor=blue!55!black,
            urlcolor=blue!55!black]{hyperref}
\usepackage{authblk}

\theoremstyle{plain}
\newtheorem{theorem}{Theorem}
\newtheorem{lemma}[theorem]{Lemma}
\newtheorem{proposition}[theorem]{Proposition}

\theoremstyle{definition}

\theoremstyle{remark}

\newcommand{\Rhat}{\widehat R}
\newcommand{\op}{\mathrm{op}}

\title{\textbf{The Universal Warmup Path:\\[2pt]
Automatic Preconditioner Selection for HMC}}
\author{Junpeng Lao}
\date{26 July 2026}

\begin{document}
\maketitle

\begin{abstract}\noindent
Euclidean Hamiltonian Monte Carlo (HMC) warmup must choose a step size and
constant preconditioner from limited, nonstationary draws. Standard warmup
follows a fixed schedule and generally requires the preconditioner structure to
be specified in advance. We present a multi-chain controller that starts
diagonal and, at dimension-derived window endpoints, selects between diagonal
and low-rank-plus-diagonal inverse mass matrices and chooses the retained rank,
subject to dimension and sample-support caps. When evidence is inconclusive,
the controller gathers another scheduled window. Persistent
within-/between-chain disagreement means draws do not support treating one
constant preconditioner as an adequate global description; the controller
retains its within-region matrix and advises a population or tempering method
for regional exploration. Poor held-out score--position linearity advises
reparameterization. The controller selected low rank in every evaluated
headline benchmark NUTS run ($12/12$).
Geometric-mean pooled ESS-per-gradient ratios relative to the prespecified
Fisher low-rank warmup and Welford diagonal warmup baselines were respectively
$2.451$ and $22.572$ on the synthetic ill-conditioned Gaussian benchmark, and
$1.951$ and $6.264$ on the German-credit Bayesian logistic-regression
posterior; all compared runs passed the post-warmup quality check.
This method unifies common HMC warmup heuristics in one evidence-driven
controller, reducing manual choices and turning warning signals into actionable
guidance.
\end{abstract}

\section{Introduction}\label{sec:intro}

Local curvature determines how finely HMC must integrate and how efficiently
it can move within a region. Global posterior structure determines whether
those locally efficient transitions explore the target adequately
\citep{duane1987hybrid,neal2011mcmc}. Euclidean HMC therefore needs a step size
and a constant mass matrix, while NUTS also adapts trajectory length
\citep{hoffman2014nuts}. Warmup must estimate the step size and preconditioner
from limited, nonstationary draws while keeping conclusions conditioned on the
regions represented in those draws.

Standard window adaptation uses a fixed schedule and generally asks the user to
choose a diagonal or dense mass-matrix structure before execution. The
controller studied here instead begins diagonal, schedules its evidence checks
from dimension, chain count, and gradient budget, and can promote from diagonal
to a low-rank-plus-diagonal constant inverse mass matrix. Without a supported
promotion it remains diagonal. Inconclusive evidence causes it to gather
another scheduled window. Persistent within-/between-chain disagreement
retains the applicable within-region matrix and produces advice to use a
population or tempering method for regional exploration. Poor held-out
score--position linearity instead produces reparameterization advice.

Previous work by \citet{bales2019selecting} constructs and selects for deployment
diagonal, dense, and Hessian-informed low-rank constant Euclidean preconditioners
from one chain's warmup window, using an $80/20$ held-out conditioning criterion
for leapfrog stability. Our finite multi-chain within/between evidence instead
gates estimator and retained-rank eligibility, yielding another window if
inconclusive or a population handoff under persistent disagreement.

For the formal analysis, $a$ indexes candidate preconditioner estimators and
$m=1,\ldots,M$ indexes chains. Each estimator maps statistics from one current
$M$-chain window to a candidate constant inverse mass matrix, whose population
target is
\[
  G_a^\star=T_a(\pi).
\]
They use the same window but different summaries: diagonal scale, pooled
per-chain-centered within-chain information, or between-chain means. The latter
two yield low-rank corrections and may share a label while targeting different
population objects. In the theorem sections, a ``route'' means a selected
estimator. During a stable route episode, the population preconditioner map may
then attract the iterate to $G_a^\star$.

Four questions remain distinct throughout. \emph{Attraction} asks whether an
estimator's fixed population update contracts. \emph{Structural
identification} asks whether the transcript resolves a rank or subspace.
\emph{Representativeness} asks whether the transcript has missed important
regions. \emph{Utility} asks whether an identified structure improves
compute-adjusted sampling. None implies the next.
Here estimator or preconditioner identification means recovering covariance
structure, rank, or subspace from a warmup transcript; it is unrelated to
causal identification or statistical-model identifiability.

``Universal'' in the title refers to using one automatic controller across the
evaluated HMC-family kernels. Under the declared constant-preconditioner
adequacy criterion, score--position nonlinearity produces reparameterization
advice, while persistent within-/between-chain disagreement produces advice to
use a population or tempering method for regional exploration. In either case,
the selected within-region preconditioner remains available for local HMC
transitions.

Our contributions are a route-indexed attractor theorem with explicit finite-
window error terms; a Markov-transcript uncertainty construction and its
operator consequences; a finite-horizon local-transcript information bound;
and an automatic HMC warmup implementation whose diagnostic threshold tests
keep identification, consistency, and utility separate. The controlled
Gaussian-mixture study holds the marginal spectrum fixed while the recorded
within/between evidence and selected estimator change, exposing both the value
and the limits of the diagnostics. Together, the fixed-geometry benchmarks test
efficiency after selection, while the mixture sweep tests whether selection
changes with the evidence. The empirical corpus compares the controller
with prespecified Fisher low-rank warmup and Welford diagonal warmup baselines,
and with a historical shared-step implementation.

Section~\ref{sec:setup} fixes the mass-matrix convention, and
Section~\ref{sec:controller} gives the practical decision tree.
Sections~\ref{sec:attractor}--\ref{sec:transcript} develop the conditional
guarantees and limits. Section~\ref{sec:efficiency} reports the empirical
comparisons; implementation, related work, and scope follow, with proofs and
supporting studies in the appendices.

\section{Mass-matrix convention and candidate estimators}\label{sec:setup}

Let $\pi(x)\propto\exp\{\mathcal L(x)\}$ on $\mathbb R^d$. We use the
\emph{inverse-mass} convention
\begin{equation}
  p\sim\mathcal N(0,G^{-1}),\qquad
  K(p)=\tfrac12p^\top Gp.                                    \label{eq:metric}
\end{equation}
Many Bayesian inference libraries, including NumPyro and BlackJAX, use
\texttt{inverse\_mass\_\allowbreak matrix} for $G=M^{-1}$, where $M$ is the HMC mass
matrix. For $x\sim\mathcal N(\mu,\Sigma)$, Hamilton's equations give
$\ddot x=-G\Sigma^{-1}(x-\mu)$, so the covariance-matching choice is
$G^\star=\Sigma$. In this convention the inverse mass matrix equals the target
covariance for a Gaussian.
All positions, scores, covariance references, and geometry statements in this
paper use the unconstrained coordinates supplied to HMC, after any
transformation from constrained model parameters.

For a general target with finite second moments, define the posterior
covariance reference
\[
  \Sigma_\pi=\operatorname{Cov}_\pi(X).
\]
It is a common reference for the HMC preconditioners studied here, not a common
estimator target: a Fisher, masked, regularized, or finite-rank estimator need
not target $\Sigma_\pi$.

For an estimator $a$, let $G_a^\star=T_a(\pi)$ and define the affine-invariant
residual
\[
  R_a(G)=G^{-1/2}G_a^\star G^{-1/2}.
\]
The tangent from $G$ toward $G_a^\star$ is represented by
$\log R_a(G)$ in the whitened frame. This is not interchangeable with the
log-Euclidean difference $\log G_a^\star-\log G$ unless the matrices commute.
For the recursion below we work in a bounded SPD chart $\theta=\log G$; the
per-estimator moment-to-preconditioner map is part of $T_a$, not silently
identified with either matrix logarithm.

Classical window adaptation updates step size in fast and slow intervals, while
learning covariance only in expanding, memoryless slow intervals
\citep{carpenter2017stan,cmdstan239warmup}. Fisher-divergence
preconditioning supplies an estimator of a constant mass matrix, or affine
preconditioner, from warmup draws and scores in diagonal, dense, or
low-rank-plus-diagonal form \citep{seyboldt2026preconditioning}. Its
mass-matrix estimation and dual-averaging step-size adaptation are separate,
and the paper also specifies its own warmup schedule. Our controller must
decide not only which estimator to use, but when the transcript makes that
estimator eligible.

\section{Automatic HMC warmup controller and evaluation protocol}\label{sec:controller}

The evaluated implementation takes a target log density and score, dimension
$d$, chain count $M$, an HMC-family kernel, and total gradient budget
$\mathcal B$. Initial positions are supplied by the caller; between-chain
evidence is informative to the extent that these starts and the subsequent
warmup expose distinct regions. Before execution the controller constructs the
complete decision schedule from $d$, $M$, the kernel cost, and $\mathcal B$. It
starts with a diagonal inverse mass matrix while adapting step size. For NUTS it assigns
$\lfloor\mathcal B/(20M)\rfloor$ transitions to each chain, using $20$ as a
conservative gradients-per-transition allowance; a fixed-$L$ kernel uses its
declared integration length in place of $20$. The dimension-derived rank
capacity and pooled support floor are
\[
 k_{\mathrm{cap}}=\min\{50,\max(\lfloor d/2\rfloor,1)\},
 \qquad N_{\min}=8(k_{\mathrm{cap}}+1).
\]
After a one-step initialization window, the first evidence-bearing window has
$n_1=\lceil N_{\min}/M\rceil$ transitions per chain. Mass-matrix windows then have
nominal $1.5\times$ growth. When another grown window would not fit before the
step-size-only phase, the last slow window instead absorbs the remaining mass-matrix
transitions. The final $15\%$ adapts step size only. These boundaries are fixed
before sampling begins.

At each prescribed mass-matrix-window endpoint, the controller pools the
current per-chain buffer across the $M$ chains, computes within-chain and
between-chain evidence, runs the implemented conjunction of diagnostic
threshold tests, and resets the buffer. The conjunction combines the
within-chain and between-means structural tests, the $R^2$
constant-preconditioner-fixability check, and enough remaining budget for the
selected rank and step-size readaptation. These structural tests are denoted W
and T, respectively, in the implementation.
The held-out score--position linearity check asks whether the observed local
geometry is adequately approximated by a constant linear preconditioner; poor
fit motivates reparameterization. Persistent within-/between-chain
disagreement instead motivates a population or tempering method for regional
exploration. The controller treats both as action signals conditioned on the
observed warmup evidence.

A supported decision promotes the diagonal matrix to a
low-rank-plus-diagonal matrix. The within-chain estimator constructs its
low-rank correction from pooled, per-chain-centered information; the
between-means estimator uses a correction based on the chain means. If no
promotion is supported, the final choice remains diagonal. Inconclusive
evidence retains the diagonal matrix and advances to the next larger scheduled
window. Persistent within-/between-chain disagreement can instead indicate that
the draws do not support treating one constant preconditioner as adequate
across the observed regions or as evidence of global exploration. In that case
the controller retains the applicable within-region matrix and reports advice
to use a population or tempering method for regional exploration. Poor
held-out score--position linearity separately supports reparameterization advice.
Decisions are evaluated at endpoints, and every window runs to its prescribed
fixed boundary.

On completion, \nolinkurl{parameters} contains the frozen
\nolinkurl{step_size} and \nolinkurl{inverse_mass_matrix}. The controller
verdict reports \nolinkurl{route} and \nolinkurl{effective_rank}; its flags
report \nolinkurl{metric_scope}, \nolinkurl{observed_ensemble_evidence}, and
\nolinkurl{handoff}. The selected constant preconditioner remains usable for
within-region HMC transitions, while a population or tempering method can
support regional exploration. Advisory outputs are recomputed at later
endpoints and may clear while the scan continues.

Estimator, window schedule, mass-matrix-buffer policy, and step-size controller
are independent engineering choices. The Fisher-HMC reference prescription in
this paper is specifically the $30/55/15$ schedule of
\citet{seyboldt2026preconditioning}: it uses $L=10$ and then $L=80$
refresh/memory periods, periodically wipes older draws, and reinitializes step
size once at the start of phase 2. Its constant-mass-matrix estimator and
dual-averaging step-size adaptation are distinct. This description is specific
to the cited Fisher-HMC prescription. Stan's expanding slow windows are
successively doubled in the version 2.39 CmdStan User's Guide
\citep{cmdstan239warmup}.

The empirical schedule study estimates performance of complete configurations
because schedule, buffer policy, and restart behavior vary together. Its
limited reseed ablation reports the observed continuous-versus-reseed
comparison. Figure~\ref{fig:schedule} in Appendix~\ref{app:schedule} records
the schedule prescriptions and one executed trace.

\begin{figure}[tbp]\centering
\resizebox{0.98\linewidth}{!}{%
\begin{tikzpicture}[
  font=\footnotesize,
  >=Stealth,
  line/.style={->,thick,draw=black!65},
  proc/.style={draw=okblue,fill=okblue!7,rounded corners=2pt,
               align=center,minimum height=8mm,text width=6.1cm},
  decision/.style={diamond,draw=black!65,fill=black!5,aspect=2.8,
                   align=center,inner sep=1.2pt,text width=4.0cm},
  select/.style={draw=okgreen,fill=okgreen!8,rounded corners=2pt,
              align=left,minimum height=13mm,text width=4.8cm},
  gather/.style={draw=black!55,fill=black!4,rounded corners=2pt,
               align=center,minimum height=13mm,text width=3.5cm},
  unsupported/.style={draw=okvermil,fill=okvermil!8,rounded corners=2pt,
                 align=center,minimum height=12mm,text width=4.4cm},
  final/.style={draw=okblue,fill=okblue!7,rounded corners=2pt,
                align=center,minimum height=8mm,text width=4.4cm}
]
\node[proc] (inputs) {Inputs: target log density/score, $d$, $M$, HMC kernel,
  and gradient budget $\mathcal B$};
\node[proc,below=4mm of inputs] (schedule) {Set
  $k_{\rm cap}=\min\{50,\max(\lfloor d/2\rfloor,1)\}$,
  $N_{\min}=8(k_{\rm cap}+1)$, and
  $n_1=\lceil N_{\min}/M\rceil$; prescribe nominal $1.5\times$ growth and the
  final $15\%$ step-size-only phase};
\node[proc,below=4mm of schedule] (window) {Start diagonal while adapting step
  size; gather the next $M$-chain mass-matrix window};
\node[decision,below=7mm of window] (boundary) {At the prescribed endpoint:
  pool buffers, evaluate within-chain and between-means evidence, check the
  deadline, reset};
\node[select,below left=8mm and 15mm of boundary] (select) {
  \textbf{Low-rank promotion supported}\\[-2pt]
  add a within-chain or between-means\\
  low-rank correction;\\
  keep the selected low-rank estimator};
\node[gather,below=8mm of boundary] (inconclusive) {
  \textbf{Inconclusive}\\[-2pt]
  retain diagonal; gather the next prescribed window};
\node[unsupported,below right=8mm and 15mm of boundary] (unsupported) {
  \textbf{Constant-preconditioner advice}\\[-2pt]
  retain the applicable within-region matrix};
\node[unsupported,below=5mm of unsupported,text width=4.4cm] (reparam) {
  poor score--position linearity:\\
  reparameterize funnels or scale coupling};
\node[unsupported,below=5mm of reparam,text width=4.4cm] (population) {
  persistent within-/between-chain disagreement:\\
  population or tempering method for regional exploration};
\node[proc,below=49mm of inconclusive] (continue) {Advance the prescribed schedule};
\node[decision,below=7mm of continue,text width=3.3cm] (exhausted)
  {Slow phase exhausted?};
\node[final,below=7mm of exhausted] (final) {Run the prescribed step-size-only
  $15\%$; then freeze step size and inverse mass matrix};
\coordinate (looplow) at ([xshift=-8mm]select.west |- exhausted.west);
\coordinate (loopup) at ([xshift=-8mm]select.west |- window.west);

\draw[line] (inputs) -- (schedule);
\draw[line] (schedule) -- (window);
\draw[line] (window) -- (boundary);
\draw[line] (boundary) -| node[pos=0.25,above left]{supported} (select);
\draw[line] (boundary) -- node[right]{inconclusive} (inconclusive);
\draw[line] (boundary) -| node[pos=0.25,above right]{unsupported} (unsupported);
\draw[line] (unsupported) -- (reparam);
\draw[line] (unsupported.east) -- ++(5mm,0) |- (population.east);
\draw[line] (select.south) |- (continue.west);
\draw[line] (inconclusive.south) -- (continue.north);
\draw[line] (reparam.east) -- ++(8mm,0) |- ([yshift=2mm]continue.east);
\draw[line] (population.east) -- ++(12mm,0) |- ([yshift=-2mm]continue.east);
\draw[line] (continue) -- (exhausted);
\draw[line] (exhausted) -- node[right]{yes} (final);
\draw[line] (exhausted.west) -- node[above]{no} (looplow)
  -- (loopup) -- (window.west);
\end{tikzpicture}%
}
\caption{The evaluated automatic HMC warmup controller. Setup fixes every
window boundary before execution, and evidence is evaluated only at
mass-matrix-window endpoints. Every outcome continues the prescribed schedule:
a supported decision promotes to a within-chain or between-means low-rank
correction, inconclusive evidence retains the diagonal matrix for the next
window, and an advisory outcome retains the applicable within-region inverse
mass matrix and reports the advice supported by its signal. Without a supported
promotion, the final matrix remains diagonal. Exhaustion of the
mass-matrix-adapting phase enters the final step-size-only $15\%$, after which
both step size and inverse mass matrix are frozen for posterior sampling;
window boundaries remain fixed throughout.}
\label{fig:warmup-path}
\end{figure}
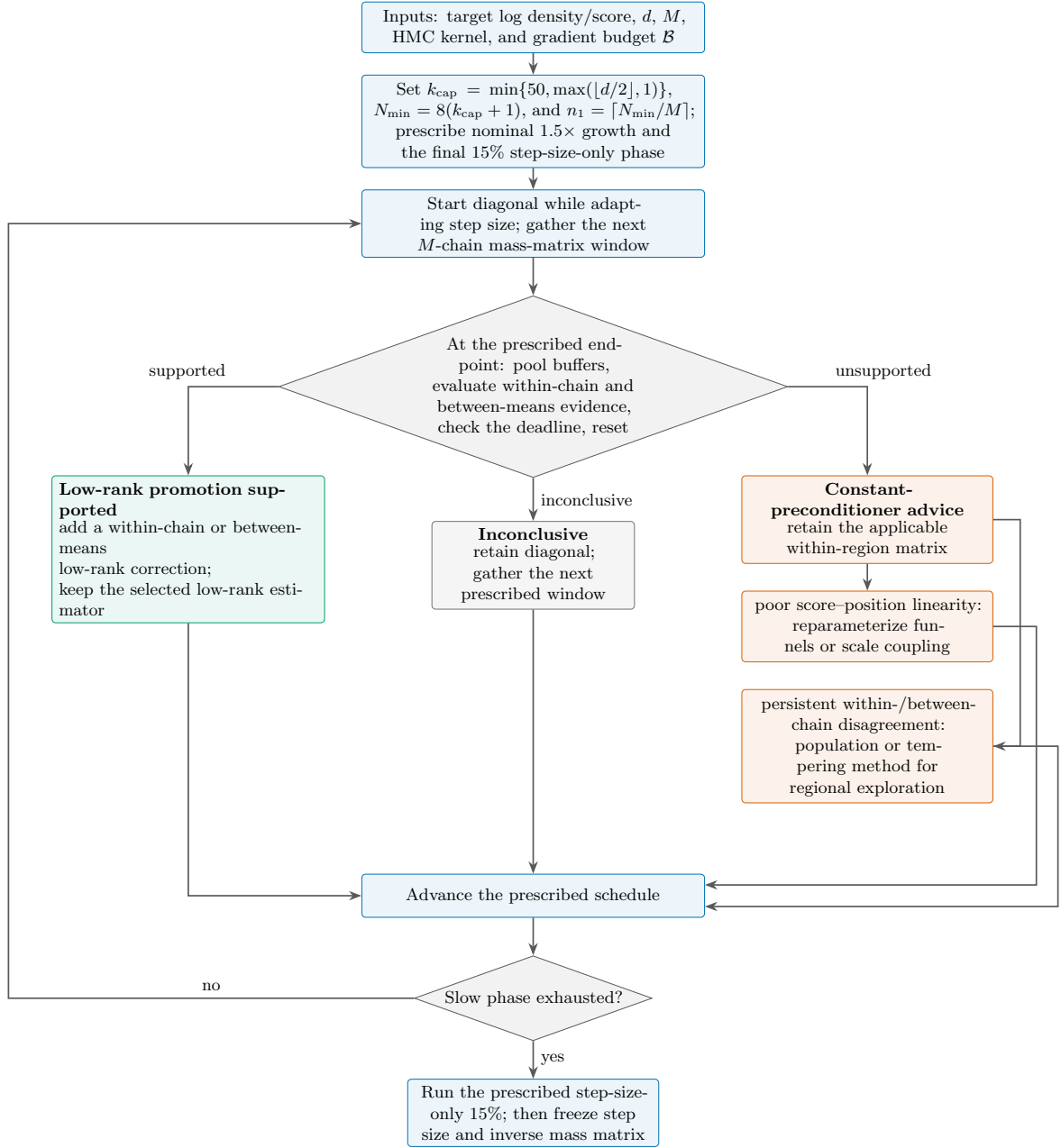

\section{Route-indexed attractor dynamics}\label{sec:attractor}

This section analyzes confidence-calibrated, no-demotion selection among
constant Euclidean inverse-mass estimators during HMC warmup. A ``route'' is
the selected estimator, and the formal ``latch'' is the rule that prevents
demotion during its route episode.

Let $\eta_k$ be the target-population routing object at the current preconditioner,
including the population moments and whitening transformation declared by the
candidate estimator. Let $C_k$ be its confidence set at the end of window $k$ and
$H_b$ the eligibility region for promotion to route $b$. Coverage of $C_k$ must
pay for starting-law, within-window adaptation, whitening, sampling, and
regularization errors. Routes are ordered and permit one transition.

\begin{lemma}[Simultaneous soundness of the confidence-aware latch]\label{lem:latch}
Let $\mathcal K$ be the fixed, declared set of potential decision windows.
Suppose
\[
  \Pr(\eta_k\notin C_k\mid\mathcal F_{k-1})\le\delta_k
\]
for every $k\in\mathcal K$, whether or not a promotion is ultimately attempted.
Promote to $b$ only when $C_k\subset H_b$. Then, with probability at least
$1-\sum_{k\in\mathcal K}\delta_k$, every promotion made by the latch is
supported by the true routing estimand at the time of promotion.
\end{lemma}

The lemma supports the declared target-population estimand at promotion time;
persistence and global target coverage are separate questions. Repeated-window
use requires confidence spending or an anytime-valid construction rather than
reusing a fixed-window Wald set.

After promotion, consider a route episode for estimator $a$. Let $s(\theta)$ be its
population routing statistic, $g$ the router, $h_a(\theta)$ the population
update field, and $\theta_a^\star=\log G_a^\star$. Write
\begin{equation}
  F_a(\theta)=\theta+\alpha h_a(\theta),\qquad
  \theta_{k+1}=\widehat F_{k,a}(\theta_k).                    \label{eq:update}
\end{equation}

\begin{theorem}[Local robust attractor of a route episode]\label{thm:attractor}
Let $\Theta$ be the bounded log-SPD chart used by the controller. Assume
$R_0>0$, $\theta_a^\star\in\Theta$, and
$\overline B_F(\theta_a^\star,R_0)\subset\Theta$, where
$\overline B_F$ denotes a closed Frobenius-norm ball. Assume also that
$s(\theta_a^\star)\in\operatorname{int}(\mathcal S_a)$, where
$\mathcal S_a=\{s:g(s)=a\}$. Suppose $h_a(\theta_a^\star)=0$ and, on
$\overline B_F(\theta_a^\star,R_0)$, for constants $\mu>0$ and $L_h>0$,
\[
 \langle\theta-\theta_a^\star,h_a(\theta)\rangle_F
 \le-\mu\|\theta-\theta_a^\star\|_F^2,\qquad
 \|h_a(\theta)\|_F\le L_h\|\theta-\theta_a^\star\|_F.
\]
Let $0<\alpha<2\mu/L_h^2$ and
$q=(1-2\alpha\mu+\alpha^2L_h^2)^{1/2}<1$. For every $k<K$, suppose
candidate statistics and updates are constructed before routing, and define
the joint error event
\[
\begin{aligned}
\mathcal E_k={}&
 \{\|\widehat s_k-s(\theta_k)\|_{\mathcal S}>r_k^s\}\\
 &{}\cup
 \{\|\widehat F_{k,a}(\theta_k)-F_a(\theta_k)\|_F>e_k\},
 \qquad
 \Pr(\mathcal E_k\mid\mathcal F_{k-1})\le\delta_k .
\end{aligned}
\]
Here $\|\cdot\|_{\mathcal S}$ is the declared statistic-space norm. Let $s$ be
$L_s$-Lipschitz from the ball to that norm, let
$\kappa=\operatorname{dist}\{s(\theta_a^\star),
\partial\mathcal S_a\}>0$, with distance in the same norm. Subject to the
ordered latch, the observable rule selects or promotes estimator $a$
whenever the following inclusion holds, and otherwise abstains:
\[
\overline B_{\mathcal S}(\widehat s_k,r_k^s)\subset\mathcal S_a.
\]
For every $k<K$, assume $L_sR_0+2r_k^s<\kappa$ and
$e_k\le(1-q)R_0$; the derived margin is sufficient for the observable
inclusion on the good event.
If $\|\theta_0-\theta_a^\star\|_F\le R_0$ and the latch is unset or already
correctly set to $a$, then on the simultaneous good event the router selects
$a$, the iterates remain in the ball, and
\begin{equation}
 \|\theta_K-\theta_a^\star\|_F
 \le q^K\|\theta_0-\theta_a^\star\|_F
     +\sum_{j=0}^{K-1}q^{K-1-j}e_j.                          \label{eq:contract}
\end{equation}
This event has probability at least $1-\sum_{k<K}\delta_k$.
\end{theorem}

The finite-window update-error budget is deliberately explicit:
\begin{equation}
 e_k\le e_k^{\rm start}+e_k^{\rm step}+e_k^{\rm white}
          +e_k^{\rm sample}+e_k^{\rm reg}.                   \label{eq:errors}
\end{equation}
These terms cover the starting law, adaptive step size, estimated whitening,
sampling fluctuation, and regularization. Controlled-Markov stability,
convergence, and ergodicity results provide asymptotic ingredients under
conditions such as drift, containment, and diminishing adaptation; they do not
themselves provide route-specific finite-window high-probability bounds for
$e_k^{\rm start}$ or $e_k^{\rm step}$. Invariant kernels alone do not remove
transient or state-dependent adaptation bias
\citep{andrieu2005stability,fort2016convergence,roberts2007coupling}.
This theory therefore analyzes a confidence-calibrated estimator-selection
rule. The evaluated controller uses fixed diagnostic thresholds; connecting it
to the theorem requires calibration of the five component bounds in
\eqref{eq:errors}.

\section{Uncertainty and safe structural decisions}\label{sec:ident}

\subsection{Posterior covariance and estimator identification}\label{sec:covbridge}

For $M$ equal-length chains with $n$ recorded states, write $\bar x_m$ for
chain $m$'s mean and $\bar x$ for the grand mean, and define
\[
 W=\frac{1}{M(n-1)}\sum_{m=1}^M\sum_{t=1}^n
 (x_{mt}-\bar x_m)(x_{mt}-\bar x_m)^\top,
\]
\[
 B=\frac{n}{M-1}\sum_{m=1}^M
 (\bar x_m-\bar x)(\bar x_m-\bar x)^\top.
\]
If $S$ is the ordinary covariance of all $Mn$ states about $\bar x$, then
\begin{equation}
 (Mn-1)S=M(n-1)W+(M-1)B.                                  \label{eq:anova}
\end{equation}
This finite-transcript identity is exact: it requires neither independence nor
stationarity. It remains exact after one common realized linear projection
$A$, with $S,W,B$ replaced by $ASA^\top,AWA^\top,ABA^\top$. A projection
selected from the same data is still algebraically descriptive, but calibrated
post-selection inference requires cross-fitting, selective inference, or a
simultaneous operator bound. Separately whitening $S$, $W$, and $B$ does not
preserve \eqref{eq:anova}.

Now hold $M$ fixed and let $n\to\infty$. Suppose each chain's empirical first
and second moments converge to those of a law $\nu_m$. Then $S$ converges to
the covariance of the equal aggregate mixture
\[
 \bar\nu=\frac1M\sum_{m=1}^M\nu_m.
\]
This limit identifies $\Sigma_\pi$ only when $\bar\nu$ is second-moment
representative of $\pi$. Agreement among chains is not enough: a clean ensemble
confined to one basin can converge to the wrong regional covariance.

\begin{figure}[tbp]\centering
\resizebox{0.98\linewidth}{!}{%
\begin{tikzpicture}[
  font=\footnotesize,
  >=Stealth,
  edge/.style={->,thick,draw=black!60},
  root/.style={draw=okblue,fill=okblue!7,rounded corners=2pt,
               align=center,minimum height=7mm,text width=4.5cm},
  outcome/.style={draw=black!55,fill=black!4,rounded corners=2pt,
                 align=center,minimum height=7mm,text width=3.7cm},
  selected/.style={draw=okgreen,fill=okgreen!8,rounded corners=2pt,
                align=center,minimum height=12mm,text width=4.25cm},
  regional/.style={draw=okvermil,fill=okvermil!8,rounded corners=2pt,
                   align=center,minimum height=12mm,text width=4.25cm}
]
\node[root] (sigma) {posterior covariance reference $\Sigma_\pi$};
\node[root,below=5mm of sigma] (split) {observed within/between split};
\node[outcome,below left=9mm and 42mm of split] (negligible)
  {between negligible};
\node[outcome,below=9mm of split] (transient) {transient};
\node[outcome,below right=9mm and 42mm of split] (persistent)
  {persistent regional};
\node[selected,below=5mm of negligible] (deploy)
  {use $W$ as the current within-chain estimate\\
   for a supported constant preconditioner\\
   {\scriptsize conditional on the observed transcript}};
\node[outcome,below=5mm of transient] (wait) {gather another scheduled window};
\node[regional,below=5mm of persistent] (ensemble)
  {retain the within-region preconditioner\\
   use a population or tempering method for regional exploration\\
   {\scriptsize population/regional companion, Paper 3}};

\draw[edge] (sigma) -- (split);
\draw[edge] (split) -- (negligible);
\draw[edge] (split) -- (transient);
\draw[edge] (split) -- (persistent);
\draw[edge] (negligible) -- (deploy);
\draw[edge] (transient) -- (wait);
\draw[edge] (persistent) -- (ensemble);
\end{tikzpicture}%
}
\caption{Practical interpretation of the posterior covariance reference through
the observed within/between split. The leaves condition use of a constant
Euclidean preconditioner on the observed transcript and identify the next
action.}
\label{fig:covariance-tree}
\end{figure}
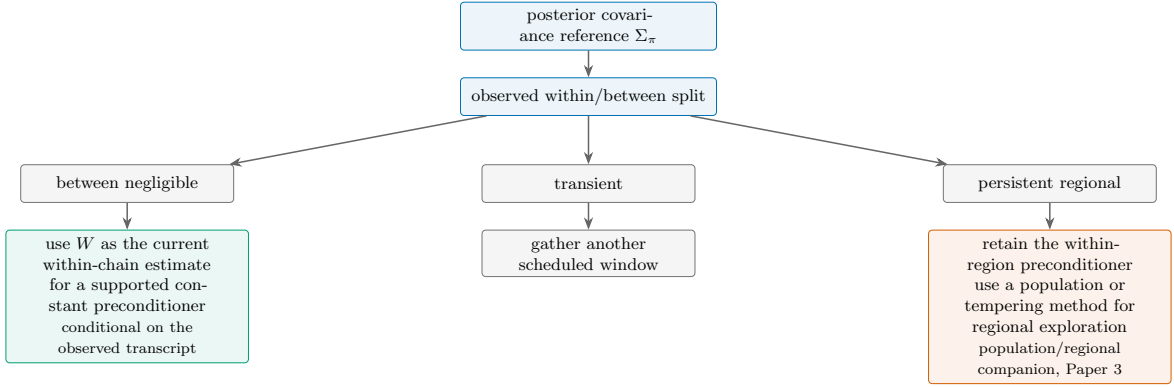

\begin{proposition}[When covariance identifies an estimator]\label{prop:fibre}
Let $\mathcal P$ be a declared target class with finite covariance. An estimator
functional $T_a$ is identified by covariance on $\mathcal P$ if and only if it
is constant on covariance fibres:
\[
 \operatorname{Cov}_{\pi_0}(X)=\operatorname{Cov}_{\pi_1}(X)
 \quad\Longrightarrow\quad T_a(\pi_0)=T_a(\pi_1)
 \qquad(\pi_0,\pi_1\in\mathcal P).
\]
Equivalently, there is a map $F_a$ on the covariance range of $\mathcal P$ such
that $T_a(\pi)=F_a(\Sigma_\pi)$. This condition preserves the estimator-indexed
target $G_a^\star=T_a(\pi)$ rather than replacing it by one common
preconditioner.
\end{proposition}

The Fisher estimator gives a useful full-rank special case. Let
$s(x)=\nabla\log\pi(x)$ and
$J_\pi=E_\pi\{s(X)s(X)^\top\}$.

\begin{proposition}[Full-rank Fisher covariance bound]\label{prop:fisher}
Suppose $\pi$ has a smooth positive density, finite SPD covariance
$\Sigma_\pi$, square-integrable score, and Stein boundary conditions. Then
\[
 E_\pi\{(X-\mu)s(X)^\top\}=-I,\qquad
 J_\pi\succeq\Sigma_\pi^{-1}.
\]
For the unregularized, no-cutoff, full-rank affine-invariant Fisher population
functional, with $\#$ denoting the SPD geometric mean,
\[
 G_F=\Sigma_\pi\#J_\pi^{-1},\qquad
 \Sigma_\pi^{-1/2}G_F\Sigma_\pi^{-1/2}
 =(\Sigma_\pi^{1/2}J_\pi\Sigma_\pi^{1/2})^{-1/2}\preceq I.
\]
Equality holds precisely in the affine-score, hence Gaussian, case under these
assumptions \citep{dembo1991information}.
\end{proposition}

This proposition does not identify the deployed finite-rank, masked,
regularized Fisher estimator from covariance. For example, a standard Gaussian and
a variance-one Student-$t_\nu$ law have the same covariance, whereas their
one-dimensional location Fisher informations are respectively $1$ and
$\nu(\nu+1)/\{(\nu-2)(\nu+3)\}>1$ for $\nu>2$. Their full-rank Fisher
functionals therefore differ.

\subsection{A per-estimator Markov central limit theorem}

For estimator $a$, collect its quadratic and score observables in a vector
$f_a(X_t)$, with mean $\eta_a$. For a stationary,
geometrically ergodic frozen-kernel chain with sufficient moments, the
multivariate Markov-chain CLT gives
\begin{equation}
 \sqrt n(\widehat\eta_a-\eta_a)
 \Rightarrow\mathcal N(0,\Omega_a),\qquad
 \Omega_a=\sum_{h=-\infty}^{\infty}
 \operatorname{Cov}_\pi\{f_a(X_0),f_a(X_h)\}.                 \label{eq:mclt}
\end{equation}
The full long-run covariance (LRV) depends on the functionals in $f_a$; a
single scalar ESS cannot be substituted into an iid Wishart law
\citep{jones2004markov,vats2018strong}.
If $\phi_a$ maps moments to a symmetric residual, let $J_a$ be the derivative
of $\operatorname{vech}\circ\phi_a$ and let
\[
 u=\operatorname{vech}(\widehat R_a-R_a),
\]
where $\operatorname{vech}$ retains one copy of each off-diagonal. The delta
method gives root-$n$ limit covariance
$\Gamma_a=J_a\Omega_aJ_a^\top$ for $\sqrt n\,u$. Let $Q$ be the oracle
population support: its orthonormal columns span
$\operatorname{range}(\Gamma_a)$. Set $z=Q^\top u$ and require
$p=\operatorname{rank}(\Gamma_a)\ge1$. If $\Gamma_a$ is singular, assume
$(I-QQ^\top)u=o_p(n^{-1/2})$, or that the residual lies exactly on the declared
support. Require
$Q^\top\widehat\Gamma_aQ$ to be positive definite and consistent on this
support. For one fixed window, an illustrative Wald construction is
\begin{equation}
 r_{\alpha,n}=
 \left\{\frac{2\chi^2_{p,1-\alpha}
 \lambda_{\max}(Q^\top\widehat\Gamma_aQ)}{n}\right\}^{1/2}.   \label{eq:radius}
\end{equation}
Indeed, for symmetric $A$,
$\|A\|_F\le\sqrt{2}\|\operatorname{vech}(A)\|_2$, which accounts
for the factor $2$ inside the squared radius and yields a Frobenius, hence
operator, radius asymptotically. The construction requires a consistent
multivariate LRV estimator and regularity for the estimator map. Estimating $Q$ or
its rank from the same window requires a separate justification. With $M$
independent frozen chains, their contributions may scale as $1/M$; chains
coupled through shared adaptation are not automatically independent.
Non-asymptotic matrix concentration is another option only under explicit
boundedness, tail, and dependence assumptions \citep{neeman2024concentration}.

\begin{theorem}[Consequences of an operator-confidence event]\label{thm:operator}
Let $R_a$ and $\widehat R_a$ be symmetric population and estimated residuals,
with eigenvalues in descending order, and suppose
$\|\widehat R_a-R_a\|_{\op}\le r$.
\begin{enumerate}
\item Weyl's inequality gives
  $|\widehat\lambda_j-\lambda_j|\le r$ for every $j$.
\item For a declared population cutoff $\tau$ and $1\le k<d$, if
  $\widehat\lambda_k-r>\tau$ and
  $\widehat\lambda_{k+1}+r<\tau$, then exactly $k$ population eigenvalues
  exceed $\tau$. Otherwise this interval test is inconclusive.
\item Let $V$ and $\widehat V$ span the matched population and estimated top-$k$
  eigenspaces (or a declared isolated eigenvalue cluster), and let $\Delta$ be
  the minimum population separation between that cluster and its complement.
  If $\Delta>2r$, then
  $\|\sin\Theta(\widehat V,V)\|_{\op}\le2r/\Delta$.
\item If $\lambda_{\min}(R_a)\ge m>r$, then
  \[
    \|\log\widehat R_a-\log R_a\|_{\op}
       \le \log\!\left(\frac m{m-r}\right)\le\frac r{m-r}.
  \]
\end{enumerate}
A preconditioner-error claim additionally requires a per-estimator
perturbation bound for the moment-to-preconditioner map $T_a$.
\end{theorem}

These are standard operator, subspace, and SPD-log consequences
\citep{bhatia1997matrix,yu2015useful,higham2008functions}. They show exactly
which margins a confidence-aware threshold test needs. The deployed cosine
statistic $\Psi$ is a deterministic agreement check between replicated
directions. High agreement is compatible with shared whitening bias and does
not establish accuracy or coverage.

\subsection{The restricted role of BBP calibration}

For iid Gaussian observations from the spiked covariance model with aspect
ratio $\gamma=d/n$, a population spike separates when
$\ell>1+\sqrt\gamma$, while the null sample-noise edge is
$(1+\sqrt\gamma)^2$ \citep{baik2005phase,baik2006eigenvalues,paul2007asymptotics}.
This asymptotic PCA calculation is a useful scale calibration. It is not an
exact edge for Markov draws, a finite-sample if-and-only-if statement, or an
unrestricted information limit. In particular, subthreshold alternatives can
retain nontrivial testing power \citep{onatski2013asymptotic}.

\begin{figure}[tbp]\centering
  \includegraphics[width=\linewidth]{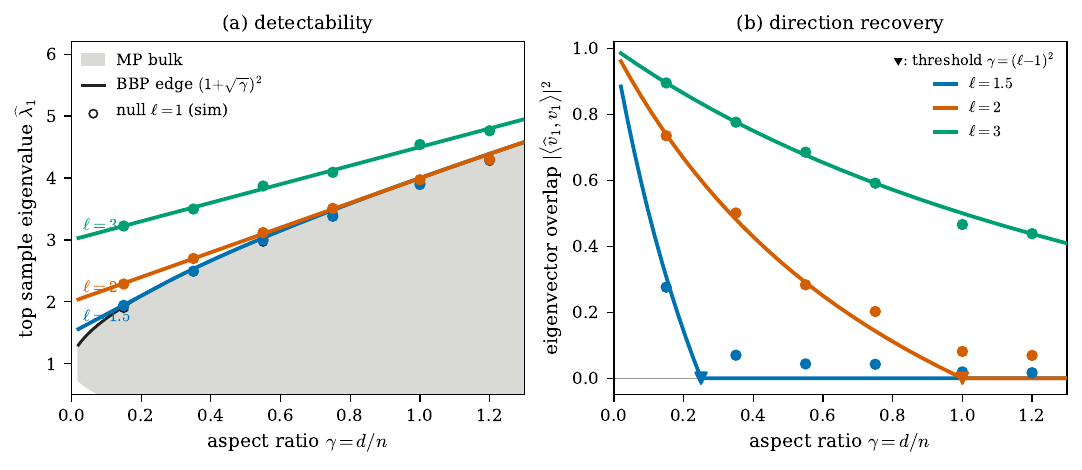}
  \caption{Iid Gaussian spiked-PCA calibration ($d=200$). The curves show the
  asymptotic sample eigenvalue and eigenvector overlap; markers are fixed-seed
  simulations. The population threshold is $1+\sqrt\gamma$, and the null sample
  edge is $(1+\sqrt\gamma)^2$. The figure is not a diagnostic threshold for
  Markov-chain draws.}
  \label{fig:bbp}
\end{figure}

\subsection{Controlled GMM at fixed marginal geometry}

The controlled target is a five-dimensional, two-component Gaussian mixture
with weights $\pi_1=0.30$, $\pi_2=0.70$,
$v=(1,1,0,0,0)^\top/\sqrt2$, and $a=21$. Its means and shared component
covariance are
\[
 \mu_1=-\pi_2\delta v,\qquad \mu_2=\pi_1\delta v,\qquad
 \Sigma_w=I_5+\beta vv^\top,
\]
where
\[
 \delta^2(\mathrm{SR})
 =\frac{\mathrm{SR}^2(1+a)}{1+\pi_1\pi_2\mathrm{SR}^2},
 \qquad \beta=a-\pi_1\pi_2\delta^2.
\]
Thus
$\Sigma_{\rm marginal}=\Sigma_w+
\pi_1\pi_2(\mu_2-\mu_1)(\mu_2-\mu_1)^\top=I_5+avv^\top$,
so the marginal diagonal-whitened top eigenvalue is $1.913$ throughout the
separation sweep. Separation is indexed by
$\mathrm{SR}=\delta/s(\delta)$, where $\delta$ is the distance between component
means and $s(\delta)=\{v^\top\Sigma_wv\}^{1/2}$ is the within-component standard
deviation along the separation direction $v$. Equation~\eqref{eq:anova}
therefore gives an exact finite-transcript partition for this target. With
latent component labels, the usual within-component plus label-mean-scatter
identity gives the corresponding population partition.

The current sweep keeps the marginal spike fixed at $1.913$ while the recorded
within/between evidence changes. The selected estimator moves increasingly
from low rank to diagonal, and an advisory population handoff appears. This
associates the estimator-selection boundary with the changing within/between
evidence rather than the fixed marginal spectrum; the recorded summaries leave
the occupation mechanism unresolved. The current-window handoff is
non-monotone. In every row, \texttt{global\_exploration} remains
\texttt{not\_established}.

Compute changes when the disagreement becomes visible: the median onset is
$5.0$ at $20{,}000$, $6.0$ at $60{,}000$, and $7.0$ at $120{,}000$ warmup
gradients. Appendix~\ref{app:gmm-boundary} gives every current $60$k row, the
single-chain contrast, and the matched ablations.

\begin{figure}[tbp]\centering
  \includegraphics[width=\linewidth]{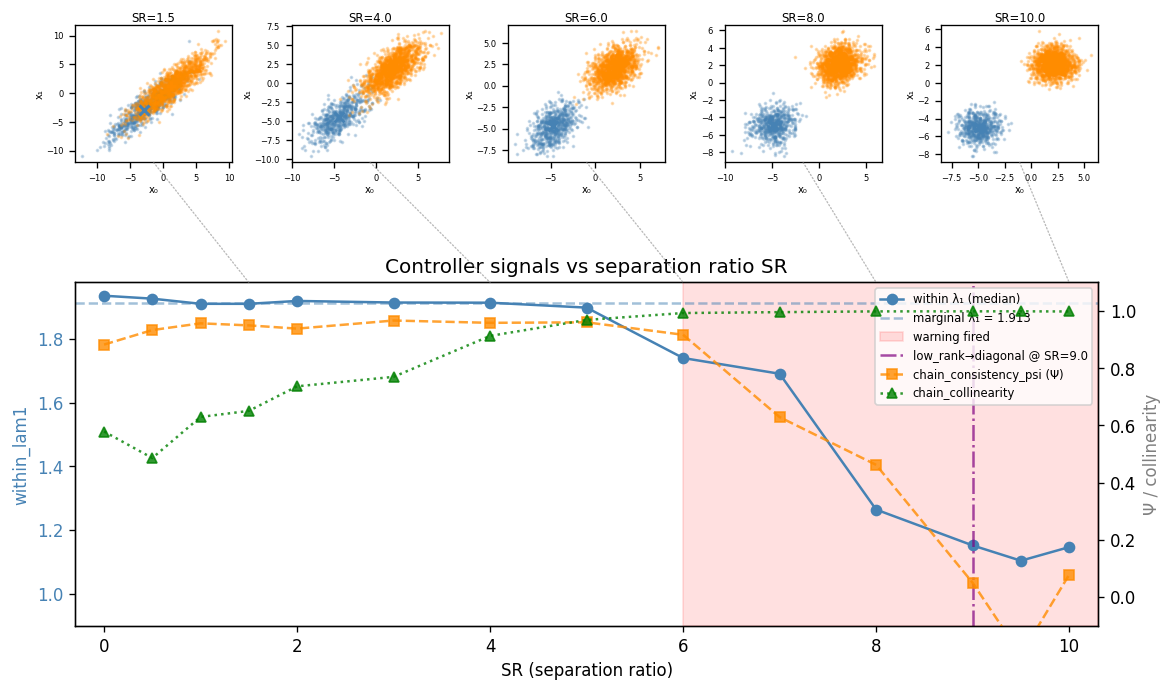}
  \caption{Current-corpus controlled GMM boundary, generated from the current
  $k=2$ main corpus. Top: target draws at selected separation ratios. Bottom:
  seed-median current diagnostics. The marginal spike remains $1.913$; the
  shading begins at the first grid point with a persistent disagreement signal
  in any seed, and the dash-dot line marks the first all-diagonal grid point.
  These are current-window signals, so a population handoff can clear.}
  \label{fig:gmm-boundary}
\end{figure}

\section{Finite-horizon local transcripts}\label{sec:transcript}

Structural evidence from visited states cannot by itself certify unseen regions.
The appropriate obstruction is finite-horizon and depends on the joint event
that every query remains local.

\begin{theorem}[Local-transcript indistinguishability]\label{thm:transcript}
Let $q_0,q_1$ be smooth, positive, integrable unnormalized targets. Target
access is only through a declared oracle, and the two targets give identical
oracle responses on an open region $A$. Run the same randomized algorithm
through a fixed horizon $T$ from the same initial law under both targets,
couple all random numbers, and assume it makes finitely many oracle calls
almost surely through $T$. Let $E_T$ be the event that every oracle query made
by the complete, possibly interacting algorithm through $T$ lies in $A$.
Induction over oracle queries then gives identical coupled transcripts on
$E_T$ and a common submeasure of mass
\[
 p_T=P_0(E_T)=P_1(E_T).
\]
If the normalized targets $\pi_0,\pi_1$ have preconditioner functionals satisfying
$d\{T(\pi_0),T(\pi_1)\}>2\epsilon$, then every common
transcript-measurable estimator $\widehat T$ satisfies
\begin{equation}
 \Pr_0\{d(\widehat T,T(\pi_0))>\epsilon\}
 +\Pr_1\{d(\widehat T,T(\pi_1))>\epsilon\}\ge p_T.            \label{eq:local}
\end{equation}
Consequently the two-point minimax error is at least $p_T/2$.
\end{theorem}

For replica events $E_{T,m}$, the chain rule is
\[
 p_T=\prod_{m=1}^M
 P(E_{T,m}\mid E_{T,1}\cap\cdots\cap E_{T,m-1}).
\]
Independent replicas give $p_T=\prod_m q_{T,m}$. If every displayed
conditional probability is bounded below by $q_T$, then
$p_T\ge q_T^M$. Identical independent replicas are sufficient for
$p_T=q_T^M$ when each has confinement probability $q_T$, but they are not
necessary: equality in the displayed lower bound occurs whenever every
conditional factor equals $q_T$, including some dependent constructions.
Shared adaptation keeps the joint $p_T$ form. The lower bound applies when
$M=1$ and makes multiple chains neither mathematically necessary nor
sufficient for discrimination. Genuinely overdispersed multiple chains remain
a principal practical safeguard against metastability and unvisited modes.
Longer runs, more starts, or a more diverse collection of starts can reduce the
confinement event, but need not do so at fixed compute.
No fixed sample count is privileged: its meaning is target-, kernel-, start-,
and controller-specific unless $p_T$ is bounded for a declared problem class.
This is a finite-horizon Le Cam two-point coupling argument
\citep{yu1997assouad}.

\section{Preconditioner selection and efficiency evidence}\label{sec:efficiency}

The headline NUTS benchmarks pair repeated seeded runs of a synthetic
ill-conditioned Gaussian with a real-data German-credit Bayesian
logistic-regression posterior. The latter is used as a posterior-geometry
benchmark and carries no causal interpretation. The implemented controller
selected low rank in all $12/12$ evaluated headline runs.

Table~\ref{tab:eff} reports observed pooled ESS-per-gradient performance for
automatic warmup relative to the prespecified Fisher low-rank warmup and
Welford diagonal warmup configurations. These roles and the nominal allocation
were preregistered rather than chosen after observing the results. Both pooled
baselines used $M=8$ chains, 312 warmup transitions per chain, and the same
proportional growing-window schedule. They differed in the mass-matrix
estimator: Fisher low rank versus Welford diagonal. The runs used float64
arithmetic and explicit integer seeds. Pooled accounting divides the effective
sample size computed from all eight chains by all warmup and sampling
gradients. Every ratio charges gradients actually consumed.

\begin{table}[htbp]\centering
\caption{Geometric-mean pooled ESS-per-gradient ratios under the declared
configurations. Ratio = automatic controller divided by prespecified baseline;
values greater than $1$ favor the automatic controller.}
\label{tab:eff}
\small
\begin{tabular}{@{}llc@{}}
\toprule
Model & prespecified warmup baseline & pooled ESS/gradient ratio \\
\midrule
ill-conditioned Gaussian & Fisher low-rank warmup & $2.451$ \\
ill-conditioned Gaussian & Welford diagonal warmup & $22.572$ \\
German-credit logistic regression & Fisher low-rank warmup & $1.951$ \\
German-credit logistic regression & Welford diagonal warmup & $6.264$ \\
\bottomrule
\end{tabular}
\end{table}

Across automatic, Fisher, and diagonal configurations, all $36/36$
post-sampling pooled-chain quality checks passed: ArviZ rank-normalized
split-$\Rhat$ \citep{vehtari2021rank} was finite with maximum at most $1.01$.
The runs recorded $4{,}996$ divergences during warmup and zero during
post-warmup sampling.

These fixed benchmarks test efficiency after automatic selection. The
controlled GMM sweep in Section~\ref{sec:ident} provides the complementary
selection test: within/between evidence changes while the marginal geometry is
held fixed.

Appendix~\ref{app:secondary-efficiency} reports two secondary estimands. The
marginal-amortized estimand treats one chain as the output and charges that
chain's sampling gradients plus one eighth of the jointly incurred warmup
gradients; its geometric means aggregate paired seed-by-chain ratios. The
one-output estimand also treats one chain as the output but charges all joint
warmup gradients. Its comparator is the historical single-chain
$2{,}500$-warmup policy, so that comparison is descriptive and not
budget-identical.

When $M$ chains share one step size, the $M$ acceptance statistics observed at a
single step size are replicated measurements for one controller time, not $M$
successive controller times. The shared-step comparison contrasts the current
mean-pooled bundle with revision
\nolinkurl{2f62921848a93e7dc544ba9de8e29ef177e373b6}, which applied those
statistics as successive dual-averaging updates. The frozen
historical/current ratios mean that the current bundle used
$19.38$--$19.89$ times fewer warmup gradients in the ill-conditioned runs and
$31.82$--$35.60$ times fewer for German credit. Mean pooling, the warmup
tree-depth cap, and the controller semantics changed together, so these ratios
describe the complete current-versus-historical configuration comparison.

In the other evaluated HMC-family kernels, low-rank selection occurred in
$12/12$ fixed-length and multinomial HMC runs with zero post-warmup
divergences. Multinomial HMC lost to its comparator in $2/6$ runs. Some
fixed-length ratios have near-zero manual denominators, so the individual
observed ratios are reported without a general performance summary.

\section{Open-source implementation}\label{sec:impl}

{\raggedright
BlackJAX provides the composable inference framework used here
\citep{cabezas2024blackjax}. The controller is available in the
\texttt{blackjax/adaptation/meta} module at merged revision
\nolinkurl{2103a4275b4d29b1650ba06458d5703eb7302b2e}.\footnote{\raggedright
\url{https://github.com/blackjax-devs/blackjax/tree/2103a4275b4d29b1650ba06458d5703eb7302b2e/blackjax/adaptation/meta}}
The frozen experiment corpus was executed from pre-merge revision
\nolinkurl{29d2468857be4de1644ca4470c2a4aa7f8137656}; its
\texttt{blackjax/adaptation/meta} tree is identical at the merged revision,
and the frozen manifests retain the executed revision as the immutable run
identity. The experiment code, frozen summary, checksummed outputs, and raw
sidecar are archived with the paper in Project Geodesic.\footnote{\raggedright
\url{https://github.com/arcueil/project-geodesic/tree/paper1-circulation-v1/papers/01-universal-warmup-path/experiments}}
\texttt{staged\_adaptation} is the host interface that supplies the step-size
controller, static schedule, and mass-matrix estimator.
\par}
\begin{lstlisting}[language=Python]
warmup = staged_adaptation(
    blackjax.nuts, logdensity_fn,
    metric="auto",
    max_grad_budget=50_000,
    n_chains=8,
)
(last_state, parameters), info = warmup.run(rng_key, initial_positions)
\end{lstlisting}
The public interface exposes one compute budget and the chain population;
Appendix~\ref{app:api} records the full signature.

\section{Related work}\label{sec:related}

Controlled-Markov stochastic approximation and adaptive-MCMC theory supply
ingredients for stability, containment, transient control, and diminishing
adaptation \citep{andrieu2005stability,fort2016convergence,roberts2007coupling}.
They do not by themselves establish the five-term finite-window update budget
in \eqref{eq:errors}. Markov CLTs and multivariate output analysis make LRV
uncertainty functional-specific
\citep{jones2004markov,vats2018strong}. Matrix concentration under dependence is
available under stronger boundedness and dependence hypotheses
\citep{neeman2024concentration}.

Window adaptation \citep{carpenter2017stan} estimates a constant Euclidean mass
matrix within a fixed schedule. \citet{seyboldt2026preconditioning} introduce a
Fisher-divergence estimator of a constant linear or affine preconditioner from
warmup draws and scores, with diagonal, dense, and
low-rank-plus-diagonal forms, together with their own warmup schedule. In that
method, mass-matrix adaptation and dual-averaging step-size adaptation remain
separate. Entropy-based adaptive HMC is another approach to mass-matrix
adaptation \citep{hirt2021entropy}.
\citet{hird2025preconditioning} characterize when constant linear
preconditioning can and cannot improve conditioning. Pathfinder instead builds
local Gaussian approximations along an L-BFGS path for variational inference
and initialization \citep{zhang2022pathfinder}. ChEES tunes HMC trajectory
length \citep{hoffman2021chees}, while MEADS adapts an ensemble of generalized
HMC chains \citep{hoffman2022meads}.

\citet{bourabee2025walnuts} use ``local adaptation'' for changing the leapfrog
step size within an HMC orbit in response to position-dependent curvature. Our
controller instead adapts between warmup windows and freezes one constant
Euclidean inverse mass matrix and one step size before posterior sampling; the
methods address within-orbit and between-window adaptation, respectively.

For regional exploration, sequential Monte Carlo supplies population-level
reweighting, resampling, and mutation over a sequence of bridging distributions
\citep{delmoral2006smc}; this does not make resampling alone a guarantee of
mode crossing. Tempered transitions are a direct regional/tempering precedent
\citep{neal1996tempered}.
The affine-invariant ensemble sampler is included as an example of robustness
to affine scaling, not as evidence of separated-mode exploration
\citep{goodman2010ensemble}.

Query-complexity lower bounds ask how many oracle calls are required to produce
a total-variation-accurate sample over declared smoothness classes
\citep{chewi2022query,he2025query}. The stochastic-gradient,
strongly-log-concave case gives a related oracle lower bound
\citep{chatterji2022oracle}. Those results differ from
Theorem~\ref{thm:transcript}, which lower-bounds estimation of a target
functional from a finite local transcript through its confinement mass $p_T$.

The corpus reports rank-normalized split-$\Rhat$ in the sense of
\citet{vehtari2021rank}. Nested-$\Rhat$ addresses a distinct design:
superchains share an initializer within groups and are conditionally
independent \citep{margossian2025nested}. It does not validate an arbitrary
post-hoc grouped rank-normalized statistic, and no such extension is used here.

The iid calibration uses spiked-PCA theory
\citep{baik2005phase,baik2006eigenvalues,paul2007asymptotics}, while
\citet{onatski2013asymptotic} cautions against turning the PCA transition into
an unrestricted testing impossibility. Weyl, Davis--Kahan variants, and matrix-
function perturbation results support only the conditional threshold-test
consequences stated in Theorem~\ref{thm:operator}
\citep{bhatia1997matrix,yu2015useful,higham2008functions}.

\section{Discussion}\label{sec:discussion}

The deployed verdict reports the selected estimator and effective rank; its
flags report conditioning scope, observed ensemble evidence, and advisory
handoff. The next implementation step is to calibrate its fixed diagnostic
thresholds to the theory's uncertainty objects.
Once our evidence-adequacy step establishes eligibility, the held-out criterion
of \citet{bales2019selecting} could rank candidate metrics; we did not evaluate
this composition.

The GMM boundary is compatible with the exact within/between decomposition but
does not identify a unique occupation mechanism. Its $k=3$ ablation exercises
the deployed low-rank structure and finds no general efficiency gain. The
finite-transcript result leaves both run length and starting diversity as ways
to reduce the finite-horizon confinement probability $p_T$, subject to the
actual interacting algorithm; it does not make multiple chains necessary.

Poor held-out score--position linearity directs the user toward
reparameterization; persistent within-/between-chain disagreement directs the
user toward a population or tempering method for regional exploration. Both
are actions based on the observed warmup evidence rather than diagnoses of the
posterior's generating mechanism. Any selected constant inverse mass matrix
remains available to an HMC method for within-region transitions. Fixed-$L$
kernels use their declared integration length for gradient-budget accounting.
The evaluated dimension/rank-capacity-indexed schedule keeps its decision
boundaries fixed.

\section{Conclusion}\label{sec:conclusion}

The Universal Warmup Path is an automatic warmup controller for the evaluated
HMC family. It begins diagonal, evaluates evidence for each estimator at
prescheduled endpoints, and either promotes to a supported
low-rank-plus-diagonal inverse mass matrix or retains the diagonal matrix. It
gathers another window when evidence is inconclusive. When the draws exhibit
persistent within-/between-chain disagreement, they do not support treating one
constant preconditioner as adequate across the observed regions. The controller
retains the applicable within-region matrix and advises a population or
tempering method for regional exploration. Poor held-out score--position
linearity separately motivates reparameterization. After a supported estimator
decision, the conditional convergence result bounds distance to that estimator's
population target by a contraction term plus accumulated finite-window
estimation error. Operator uncertainty and
local-transcript limits govern when that decision is available. In the
evaluated headline benchmark NUTS runs, the controller achieved higher observed
pooled ESS per gradient than both prespecified warmup baselines; all compared
runs passed the declared post-warmup quality check. Schedule performance
remained configuration-dependent. The method is scoped to between-window
selection of a constant Euclidean inverse mass matrix and step size for HMC.
The next step is calibration of estimator-selection uncertainty and integration
with a population or tempering companion when persistent regional disagreement
indicates a need for regional exploration.

\section*{Generative-AI use disclosure}

Anthropic Claude Opus 4.8 and Claude Fable 5, and OpenAI Codex and GPT-5.6,
supported literature discovery and checking, proof critique and derivation,
code generation and review, experiment orchestration and validation, and
manuscript drafting and editing. Junpeng Lao directed the work, reviewed and
revised model output, checked model-supported claims against cited primary
sources, source code, and frozen experiment outputs, and accepts responsibility
for the manuscript. No AI system is an author or a source.

\bibliography{references}

\appendix
\section{Proofs and technical details}\label{app:proofs}

\begin{proof}[Proof of Lemma~\ref{lem:latch}]
On the event $\eta_k\in C_k$, the inclusion $C_k\subset H_b$ implies
$\eta_k\in H_b$. Iterated expectation and a union bound over the fixed
potential-window set $\mathcal K$ give
$\Pr(\exists k\in\mathcal K:\eta_k\notin C_k)
\le\sum_{k\in\mathcal K}\delta_k$. This bound is taken before observing which
windows trigger an opportunity. Every realized promotion is therefore
supported on the complementary event.
\end{proof}

\begin{proof}[Proof of Proposition~\ref{prop:fibre}]
If $T_a=F_a\circ\operatorname{Cov}$, equal covariance plainly implies equal
estimator target. Conversely, if $T_a$ is constant on every covariance fibre,
define $F_a(\Sigma)$ to be $T_a(\pi)$ for any $\pi\in\mathcal P$ with
$\operatorname{Cov}_\pi(X)=\Sigma$. Fibre constancy makes this definition
independent of the representative.
\end{proof}

\begin{proof}[Proof of Proposition~\ref{prop:fisher}]
Integration by parts under the Stein boundary conditions gives
$E\{(X-\mu)s(X)^\top\}=-I$. Hence the covariance matrix of
$(X-\mu,s(X))$ is
\[
 \begin{pmatrix}\Sigma_\pi&-I\\-I&J_\pi\end{pmatrix}\succeq0.
\]
Its Schur complement gives $J_\pi\succeq\Sigma_\pi^{-1}$. The definition of
the SPD geometric mean then yields
\[
 \Sigma_\pi^{-1/2}(\Sigma_\pi\#J_\pi^{-1})\Sigma_\pi^{-1/2}
 =(\Sigma_\pi^{1/2}J_\pi\Sigma_\pi^{1/2})^{-1/2}\preceq I.
\]
Equality in the Schur complement makes
$s(X)=-\Sigma_\pi^{-1}(X-\mu)$ almost surely. Integrating this affine score
identifies the Gaussian density; the Gaussian attains equality. For the
variance-one Student-$t_\nu$ example, direct integration of the squared
location score gives
$J_\pi=\nu(\nu+1)/\{(\nu-2)(\nu+3)\}$.
\end{proof}

\begin{proof}[Proof of Theorem~\ref{thm:attractor}]
Write $d_k=\theta_k-\theta_a^\star$. Dissipativity and the growth bound
give
\[
 \|d_k+\alpha h_a(\theta_k)\|_F^2
 \le(1-2\alpha\mu+\alpha^2L_h^2)\|d_k\|_F^2.
\]
For any $z\in\overline B_{\mathcal S}(\widehat s_k,r_k^s)$, the route
inequalities give
\[
\|z-s(\theta_a^\star)\|_{\mathcal S}
\le2r_k^s+L_s\|\theta_k-\theta_a^\star\|_F<\kappa.
\]
Because $s(\theta_a^\star)$ is interior to $\mathcal S_a$, this proves
$\overline B_{\mathcal S}(\widehat s_k,r_k^s)\subset\mathcal S_a$. The
candidate update is therefore
selected without conditioning its error event on the route.
Equation~\eqref{eq:update} then gives
$\|d_{k+1}\|_F\le q\|d_k\|_F+e_k$. The basin condition closes the induction,
iteration yields \eqref{eq:contract}, and a conditional union bound gives the
stated probability.
\end{proof}

\begin{proof}[Proof of Theorem~\ref{thm:operator}]
The eigenvalue intervals are Weyl's inequality. The two strict inequalities
around $\tau$ certify exactly $k$ eigenvalues above the cutoff. Applying the
Davis--Kahan sin-$\Theta$ theorem on an event with gap exceeding $2r$ gives the
displayed conservative bound. Finally, the spectral floor keeps both matrices
SPD along the perturbation path; the integral representation of the matrix
logarithm bounds its derivative by the reciprocal spectral floor, giving
$\log\{m/(m-r)\}\le r/(m-r)$. None of these steps specifies the map from the
residual moments to an estimator's preconditioner, hence the separate $T_a$
condition.
\end{proof}

\begin{proof}[Proof of Theorem~\ref{thm:transcript}]
The initial states and randomized choices are coupled. If the first $j-1$
queries lie in $A$, their oracle responses agree under $q_0$ and $q_1$, so the
algorithm states and the $j$th query agree. Induction over all oracle calls
therefore makes the complete transcripts identical on $E_T$ and proves
the common submeasure of mass $p_T$. On that submeasure the estimator
output is the same. The two $\epsilon$-balls around $T(\pi_0)$ and $T(\pi_1)$
are disjoint, so that output must be wrong for at least one target. Integrating
over the common mass yields \eqref{eq:local}; taking the larger error gives the
minimax bound.
\end{proof}

\section{Secondary efficiency estimands}\label{app:secondary-efficiency}

The marginal-amortized estimand divides one chain's effective sample size by
that chain's sampling gradients plus one eighth of the warmup gradients shared
by the eight chains. It uses the same pooled Fisher low-rank warmup and Welford
diagonal warmup baselines described in Section~\ref{sec:efficiency}. The
one-output estimand divides one chain's effective sample size by that chain's
sampling gradients plus all shared warmup gradients. It compares with
historical single-chain versions of those estimator families and is descriptive
rather than budget-identical.

\begin{table}[htbp]\centering
\caption{Secondary geometric-mean ESS-per-gradient ratios. Ratio = automatic
controller divided by the listed comparator; values greater than $1$ favor the
automatic controller. The one-output comparison is descriptive and not
budget-identical.}
\label{tab:secondary-eff}
\small
\begin{tabular}{@{}llcc@{}}
\toprule
Model & warmup baseline family & marginal-amortized & one-output \\
\midrule
ill-conditioned & Fisher low-rank warmup & $2.324$ & $1.409$ \\
ill-conditioned & Welford diagonal warmup & $21.124$ & $14.196$ \\
German credit & Fisher low-rank warmup & $1.882$ & $1.131$ \\
German credit & Welford diagonal warmup & $5.887$ & $6.009$ \\
\bottomrule
\end{tabular}
\end{table}

\section{Controlled GMM boundary study}\label{app:gmm-boundary}

The target construction pins the marginal diagonal-whitened top eigenvalue at
$1.913$. Here $\mathrm{SR}=\delta/s(\delta)$ is the separation ratio defined in
\S\ref{sec:ident}. The main $60$k sweep is:
\begin{table}[htbp]\centering
\caption{Current $60$k $k=2$ sweep across three seeds. ``Persistent'' counts
\nolinkurl{persistent_disagreement_signal}; ``handoff'' counts the advisory
\texttt{population} endpoint.}
\label{tab:gmm-regimes}
\small
\setlength{\tabcolsep}{5pt}
\begin{tabular}{@{}rcccc@{}}
\toprule
SR & low-rank/diagonal & persistent & handoff & within $\lambda_1$ range \\
\midrule
$0$   & $3/0$ & $0/3$ & $0/3$ & $1.910$--$1.979$ \\
$0.5$ & $3/0$ & $0/3$ & $0/3$ & $1.916$--$1.932$ \\
$1$   & $3/0$ & $0/3$ & $0/3$ & $1.907$--$1.952$ \\
$1.5$ & $3/0$ & $0/3$ & $0/3$ & $1.907$--$1.962$ \\
$2$   & $3/0$ & $0/3$ & $0/3$ & $1.895$--$1.937$ \\
$3$   & $3/0$ & $0/3$ & $0/3$ & $1.882$--$1.929$ \\
$4$   & $3/0$ & $0/3$ & $0/3$ & $1.913$--$1.920$ \\
$5$   & $3/0$ & $0/3$ & $0/3$ & $1.842$--$1.907$ \\
$6$   & $3/0$ & $2/3$ & $2/3$ & $1.737$--$1.847$ \\
$7$   & $3/0$ & $3/3$ & $2/3$ & $1.371$--$1.707$ \\
$8$   & $2/1$ & $2/3$ & $2/3$ & $1.205$--$1.266$ \\
$9$   & $0/3$ & $2/3$ & $2/3$ & $1.093$--$1.297$ \\
$9.5$ & $0/3$ & $2/3$ & $2/3$ & $1.102$--$1.258$ \\
$10$  & $1/2$ & $2/3$ & $2/3$ & $1.081$--$1.219$ \\
\bottomrule
\end{tabular}
\end{table}

\paragraph{Single-chain contrast.}
At the evaluated budget and seed, one chain selects low rank at
$\mathrm{SR}=1.5$, diagonal at $\mathrm{SR}=5$, and diagonal at
$\mathrm{SR}=10$. The last row has mode-weight estimate $1.0$. Ensemble
evidence is unassessed for these single-chain rows. The outcomes are specific
to the implemented diagnostic thresholds and evaluated budget.

\paragraph{Matched diagonal ablation.}
The comparison holds warmup, step size, initialization, and $M=8$ fixed. Its
low-rank/matched-diagonal projection bulk ESS($v^\top x$) per gradient ratios
are $0.869230$, $0.997721$, $1.170288$, and $0.997752$ at
$\mathrm{SR}=1.5,4,5,9$, respectively. These four evaluated comparisons are
descriptive: identification of anisotropy alone does not establish a general
efficiency gain.

\paragraph{Three-axis main sweep.}
For $k=3$, all three seeds report persistent disagreement from the first
evaluated $\mathrm{SR}=6$ grid point onward. The endpoint population handoff is
nevertheless $2/3$ at each of those separation ratios and can clear within a
run. The persistent evidence field and current-window handoff therefore answer
different questions.

\paragraph{Three-axis matched ablation (null).}
A same-$M$ paired ablation used $k=3$ correlated axes\footnote{For
$v$ supported equally on $k$ axes, diagonal whitening gives
$\lambda_1(k)=(1+a)/(1+a/k)$. For $k=2,\ldots,5$, the marginal spikes are
$\{1.913,2.750,3.520,4.231\}$ and the off-diagonal counts are
$\{1,3,6,10\}$, respectively. For $k=2$, the diagnostic-threshold implementation's
asymptotic covariance-spike comparison is $1.913<2$, where $2$ is its
escalation threshold; this is neither a general impossibility nor a
finite-budget guarantee. At $k=3$, $2.750>2$, so that particular threshold
contrast is absent.} at
$\mathrm{SR}\in\{3.0,3.5,4.0,4.5\}$, seeds $1001$--$1016$, and $M=8$.
Both arms shared the $60$k warmup, terminal state, and step size; each then
produced $4{,}000$ draws per chain. For equal-SR, within-seed pairs, the
low-rank/matched-diagonal projection bulk ESS($v^\top x$) per gradient
geometric-mean ratio was $1.037$; the seed-cluster $t$ $95\%$ interval was
$[0.998,1.077]$. Because the interval includes $1$, the stated
measurable, non-noise go/no-go criterion fails: this is a null result and
supports no efficiency claim. Per-SR geometric means were $1.002$,
$0.999$, $0.987$, and $1.169$ in ascending SR order; the
$\mathrm{SR}=4.5$ rise is local and does not generalize. All $64/64$ runs
completed, with $16/16$ valid low-rank profiles at each SR, zero divergences,
and maximum projection split-$\Rhat$ $1.006$. Of the $64$ paired ratios,
$39$ are above $1$ and $25$ below.

\section{Schedule-comparison scope}\label{app:schedule}

\begin{figure}[tbp]\centering
  \includegraphics[width=\linewidth]{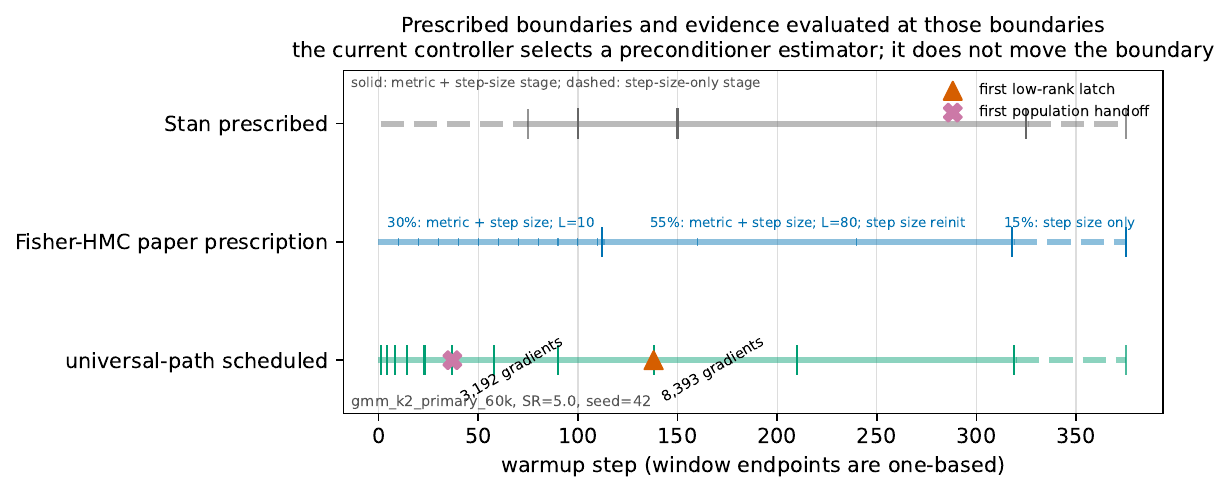}
  \caption{Schedule evidence and one executed controller trace. The Stan and
  Fisher-HMC lanes are unexecuted reference prescriptions; Fisher-HMC denotes
  the specific $30/55/15$ prescription in \citet{seyboldt2026preconditioning}.
  Controller estimator decisions are evaluated at prescribed boundaries, which
  remain fixed throughout. In the displayed trace the first
  population-handoff marker is transient and clears.}
  \label{fig:schedule}
\end{figure}

\paragraph{Full factorial.}
The validated schedule-by-buffer experiment records poor radon performance for
the \emph{bundle} combining proportional-growing windows with an
accumulating/per-draw-recompute buffer. Because accumulation and recomputation
move together in this arm, the observation estimates performance of that
complete bundle. The remaining configurations show the joint performance of
schedule, buffer, and controller semantics. The controller boundaries in
Figure~\ref{fig:schedule} are prescribed before execution and remain fixed
through the run.

\paragraph{Restart ablation.}
Continuous dual averaging has mixed efficiency: it wins $3/6$ evaluated runs, with
continuous/reseed ratios from $0.9464$ to $1.1385$ and geometric means $1.014$
for the ill-conditioned target and $1.024$ for German credit. It reduces the
recorded oscillation and number of warmup divergences; these are the observed
configuration-level results of the limited reseed ablation.

\clearpage
\section{Full warmup signature}\label{app:api}

\begin{lstlisting}[language=Python,basicstyle=\ttfamily\footnotesize,
                   aboveskip=2pt,belowskip=2pt]
staged_adaptation(
  algorithm, logdensity_fn, metric="welford_diag" | MetricRecipe | MetricCore,
  *, max_grad_budget=None, n_chains=1, imm_shrinkage_to_previous=0.0,
  initial_inverse_mass_matrix=None, initial_step_size=1.0,
  target_acceptance_rate=0.80, adaptation_info_fn=..., integrator=velocity_verlet,
  schedule_fn=None, initial_metric_state=None, **extra_parameters)
\end{lstlisting}
{\small Step-size dual averaging and the stage schedule live in the host;
mass-matrix estimation is delegated to \texttt{metric}.\par}

\end{document}